\documentclass[9pt]{article}
\usepackage{spconf}
\usepackage{CJK}
\usepackage{amsfonts}
\usepackage{slashbox}
\usepackage{amsmath}
\usepackage{multirow}
\usepackage{graphicx}
\usepackage{amsfonts}
\usepackage{amsmath,epsfig}
\usepackage{mathbbold}
\usepackage{stfloats}
\usepackage{amssymb}
\usepackage{url}
\usepackage{bm}
\usepackage{cite}
\usepackage{amsmath}
\usepackage{amssymb}
\usepackage{latexsym}
\usepackage{multirow}
\usepackage{epsfig}
\usepackage{graphics}
\usepackage{mathrsfs}
\usepackage{algorithmic}
\usepackage{algorithm}
\usepackage{subfigure}
\usepackage{slashbox}
\usepackage[all]{xy}

\usepackage{pifont}
\usepackage{bbding}
\usepackage{stmaryrd}
\usepackage{amssymb}
\usepackage{amsfonts}
\usepackage{epic}
\usepackage{stfloats}
\usepackage{epstopdf}
\usepackage{bm}
\usepackage{xcolor}
\usepackage{mathrsfs}
\usepackage{pifont}
\usepackage{bbding}
\usepackage{amsmath,epsfig}
\usepackage{mathbbold}
\usepackage{stmaryrd}
\usepackage{amsmath}
\usepackage{amssymb}
\usepackage{amsthm}

\usepackage{amsfonts}
\usepackage{epic}
\usepackage{graphicx}
\usepackage{subfigure}
\usepackage{enumerate}
\usepackage{latexsym}
\usepackage{epstopdf}
\usepackage{multirow}
\usepackage{stfloats}
\usepackage{bm}
\usepackage{booktabs}
\usepackage{color}
\usepackage{cases}

\newtheorem{proposition}{Proposition}

\newcommand{\A}{\bm A}

\newcommand{\G}{\bm G}

\newcommand{\ttheta}{\mathbf \Theta}

\usepackage{geometry}
\begin{document}
\title{Beamforming Optimization for Intelligent Reflecting  Surface with Discrete  Phase Shifts }
\name{Qingqing Wu and Rui Zhang}
 \vspace{-3mm}
\address{Department of Electrical and Computer Engineering, National University of Singapore, Singapore \\
    Email:\{elewuqq, elezhang\}@nus.edu.sg}

\maketitle

\begin{abstract}
Intelligent reflecting surface (IRS) is  a cost-effective solution  for achieving high spectrum and energy efficiency in future wireless communication systems  by leveraging massive low-cost passive elements that are able to reflect the signals with adjustable phase shifts. Prior  works on IRS mostly consider continuous phase shifts at each reflecting element, which however, is practically difficult to realize due to the hardware limitation.  In contrast, we study in this paper an IRS-aided wireless network, where an IRS with only a finite number of phase shifts at each element  is deployed to assist in  the communication from a multi-antenna access point (AP) to a single-antenna user. We aim to minimize the transmit power at the AP by jointly optimizing the continuous  transmit beamforming at the AP and discrete reflect beamforming  at the IRS, subject to a given signal-to-noise ratio (SNR) constraint at the user receiver. We first propose a suboptimal and  low-complexity solution to the problem  by exploiting the alternating optimization technique. Then, we analytically show that as compared to the ideal  case with continuous phase shifts,  the IRS with discrete phase shifts achieves the same squared power gain with asymptotically large number of reflecting elements, while a constant performance loss is incurred that depends only on the number of phase-shift levels.
Simulation results verify our analytical result as well as the effectiveness of our proposed design as compared to different benchmark schemes.
\end{abstract}

\begin{keywords}
Intelligent reflecting surface, passive array, beamforming, discrete phase shifts.
\end{keywords}

\section{Introduction}
\vspace{-0.2cm}
Although massive multiple-input multiple-output (MIMO) technology has significantly improved the spectrum and energy efficiency of wireless communication systems, the required high complexity and high  hardware cost is still  the main hindrance  to its implementation in practice, especially at higher frequencies such as those in the millimeter-wave (mmWave) band   \cite{Hien2013,foad16jstsp}.  Recently, intelligent reflecting surface (IRS) has emerged as a new and  cost-effective solution for achieving  high beamforming and/or interference suppression gains  via only low-cost reflecting elements. An IRS is generally composed of a large number of passive elements each able to reflect the incident signal with an adjustable phase shift. By intelligently tuning the phase shifts of all elements adaptive to dynamic wireless channels,  the  signals reflected by the IRS can add constructively or destructively with non-reflected signals  at the user receiver to boost the received  signal power or suppress the co-channel interference, thus drastically enhancing the wireless network performance without the need of deploying additional active transmitters/relays.

\begin{figure}[!t]
\centering
\includegraphics[width=0.4\textwidth]{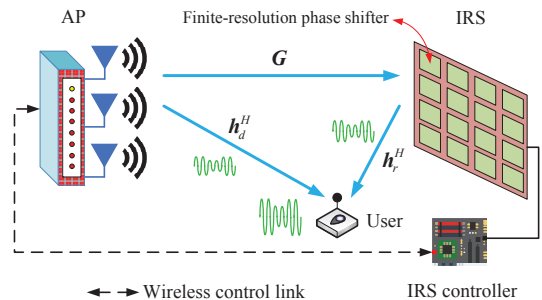}
\vspace{-0.4cm}
\caption{An IRS-aided  wireless system. } \label{system:model}\vspace{-0.6cm}
\end{figure}




Prior works  on IRS-aid wireless communication can be found in e.g.  \cite{wu2018IRS,JR:wu2018IRS,huangachievable,tan2016increasing,tan2018enabling}.
Specifically, for the IRS-aided wireless system with a single user,  it was shown in  \cite{wu2018IRS} that the IRS is capable of creating a ``signal hotspot'' in its vicinity via joint active beamforming at the access point (AP) and passive beamforming at the IRS. In particular, an asymptotic receive power gain in the order of  $\mathcal{O}(N^2)$ in terms of the number of reflecting elements at the IRS, $N$, as $N\rightarrow \infty$ was revealed in \cite{wu2018IRS}. Such a squared  power gain is larger than that of massive MIMO, i.e., $\mathcal{O}(N)$  \cite{Hien2013}, which is fundamentally due to the fact that the IRS combines the  functionalities of both receive and transmit antennas, thus doubling the gain.  For multiuser systems, it was shown in \cite{JR:wu2018IRS,huangachievable} that besides enhancing the desired signal power at the receiver,  a nearly ``interference-free'' zone can be established in the proximity of the IRS, thanks to its spatial interference nulling/cancellation  capability. However, all the above benefits are achieved by assuming an IRS with continuous  phase shift at each reflecting element  which is practically costly to implement due to the hardware limitation  \cite{tan2016increasing,tan2018enabling}.
  Although  \cite{tan2016increasing,tan2018enabling} considered the use of finite-level phase shifters for IRS, the optimal phase shifts at all elements were  obtained by exhaustive search, which, however,  is computationally prohibitive for practical IRS with a large number of elements. Thus, more efficient discrete-phase reflect beamforming for the IRS jointly with the continuous transmit beamforming for the AP needs to be designed. In addition, the performance gap between the ideal case of continuous phase shifts and the practical case with arbitrary number of discrete phase shift levels has yet to be investigated for the IRS.

Motivated by the above, in this paper we consider an IRS-aided wireless communication system as shown in Fig. \ref{system:model}, where  a multi-antenna AP serves a single-antenna user with the help of an IRS.  Although the system setup is same as that in \cite{wu2018IRS}, we consider the practical case where  the IRS only has a finite number of discrete  phase shifts in contrast to the continuous phase shifts considered in \cite{wu2018IRS}.  Similar to \cite{wu2018IRS}, we aim to  minimize the transmit power required at the AP via jointly optimizing the  active transmit  beamforming at the AP and passive reflect beamforming at the IRS (with discrete phase shifts), subject to a given signal-to-noise ratio (SNR) constraint at the user receiver. As this problem is  non-convex, we propose a low-complexity algorithm to solve it sub-optimally  by leveraging the alternating optimization (AO) technique. Specifically, the optimal discrete phase shifts of all elements are determined one by one in an iterative manner with those of the others being fixed.
 Moreover,  we analytically show that when the number of reflecting elements of the IRS, $N$, increases, the performance loss with discrete phase shifts  from that with continuous phase shifts is a constant that depends only on the number of phase-shift levels at each element, but regardless of $N$ as $N\rightarrow \infty$.
As a result, the asymptotic squared power gain of the IRS shown in \cite{wu2018IRS} with continuous phase shifts still holds with discrete phase shifts.
  Simulation results validate our analysis and also demonstrates the significant power saving at the AP by using IRS even with discrete phase shifts.
%

\vspace{-0.2cm}
\section{System Model}
\vspace{-0.2cm}
As shown in Fig. \ref{system:model}, we consider a multiple-input single-output (MISO) wireless system where an IRS composed of $N$ reflecting elements is deployed to assist in the communication from an AP with $M$ antennas to a single-antenna user. While this paper focuses on the downlink communication, the results and analysis are applicable to the uplink communication as well.  In practice, the IRS is attached  with a controller which communicates with the AP via a separate wireless link for coordination and exchanging information on channel knowledge and accordingly adjusts the phase shifts of all elements \cite{tan2016increasing,subrt2012intelligent}.
Due to the substantial path loss, we only consider one single signal reflection by the  IRS and ignore the signals that are reflected by the IRS two or more times.
 In addition, we assume a quasi-static flat-fading model for all the channels involved.

Denote by $\bm{h}^H_d\in \mathbb{C}^{1\times M}$,  $\bm{h}^H_r\in \mathbb{C}^{1\times N}$, and $\bm{G}\in \mathbb{C}^{N\times M}$, the baseband equivalent channels of the  AP-user link, IRS-user link, and AP-IRS link, respectively, where the superscript  $H$ denotes the conjugate transpose operation and
$\mathbb{C}^{x\times y}$ represents a $x\times y$ complex-valued matrix.
Note that the AP-IRS-user channel  is typically referred to as dyadic backscatter channel in the literature \cite{griffin2009complete}, which  resembles a keyhole/pinhole propagation \cite{paulraj2003introduction}. Specifically, each element at the IRS first combines all the received multi-path signals, and then re-scatters the combined signal with a certain phase shift as if from a point source, thus leading to a ``multiplicative'' channel model.  
Let  $\ttheta  = \text{diag} (\beta e^{j\theta_1}, \cdots, \beta e^{j\theta_n}, \cdots,  \beta e^{j\theta_N})$, with $\text{diag}(\mathbf{a})$ denoting a diagonal matrix with its diagonal elements given in the vector $\mathbf{a}$ and  $j$  representing the imaginary unit,  denote the phase-shift matrix of the IRS, where $\theta_n\in [0, 2\pi]$ and $\beta \in [0, 1]$ are the phase shift and amplitude reflection coefficient of each element, respectively.
 In practice, it is usually desirable to maximize the signal reflection by the IRS. Thus, for simplicity,  we set $\beta = 1$ in the sequel of this paper.
   For ease of practical implementation, we consider that the phase shift at each element of the IRS can only take a finite number of discrete values, which are equally spaced in $[0, 2\pi)$. Denote by $b$ the number of bits used to represent each of the levels. Then the set of phase shifts at each element is given by $\mathcal{F}= \{0,\Delta\theta, \cdots, \Delta\theta( {K}-1) \}$ where $\Delta\theta= 2\pi/K$ and $K=2^b$.

At the AP,  we consider the conventional continuous transmit  beamforming with  $\bm{w}\in \mathbb{C}^{M\times 1}$  denoting the transmit beamforming vector. The total transmit power is given by $\|\bm{w}\|^2$, where  $\|\cdot\|$ denotes the Euclidean
norm of  a complex vector. The signal directly from the AP and that reflected by the IRS are combined at the user receiver, i.e.,
\begin{align}\label{SectionII:channel}
y (\ell)= ( \bm{h}^H_r\Theta \bm{G} +  \bm{h}^H_d)\bm{w}s(\ell) + z(\ell), \ell =1, 2, \cdots,
\end{align}
where  $\ell$ denotes the symbol index, $s(\ell)$'s denote the information-bearing symbols which are modeled as independent and identically distributed (i.i.d.) random variables with zero mean and unit variance, and $z(\ell)$'s denote i.i.d. additive white Gaussian noise (AWGN) at the receiver with zero mean and variance $\sigma^2$. Accordingly, the user receive SNR  is given by
\begin{align}\label{SectionII:receivedpower}
\rho ={ {|( \bm{h}^H_r\Theta \bm{G}+\bm{h}^H_d)\bm{w} |^2}}/{\sigma^2}.
\end{align}
\vspace{-0.9cm}
\section{Problem Formulation}
\vspace{-0.2cm}
Let $\gamma$ denote the SNR requirement of the user and $\bm{\theta}= [\theta_1, \cdots, \theta_N]$.  In this paper, we aim to minimize  the total transmit power at the AP by jointly optimizing the transmit beamforming $\bm{w}$ and phase shifts $\bm{\theta}$, subject to the SNR constraint as well as discrete phase-shift constraints.
The corresponding optimization problem is formulated as
\begin{align}
\text{(P1)}: ~~\min_{\bm{w}, \bm{\theta}} ~~~& \|\bm{w}\|^2  \label{eq:obj}\\
\mathrm{s.t.}~~~~&| (\bm{h}^H_r\ttheta \bm{G}+\bm{h}^H_d )\bm{w}|^2   \geq \gamma \sigma^2,   \label{SINR:constraints} \\
&\theta_n \in \mathcal{F}, \forall n. \label{phase:constraints}
\end{align}
Problem (P1) is a non-convex since the left-hand-side (LHS) of \eqref{SINR:constraints} is not jointly concave with respect to $\bm{w}$ and $\bm{\theta}$, and the constraints in \eqref{phase:constraints} restrict $\theta_n$'s  to be discrete values. In general, there is no standard method for obtaining the optimal solution to such a non-convex optimization problem efficiently. An exhaustive search over all possible combinations of discrete phase shifts at all elements incurs an exponential complexity of order $O(2^{bN})$, which is prohibitive for practical systems with large $N$.

\vspace{-0.2cm}
\section{Proposed Algorithm}
\vspace{-0.2cm}
In this section, we propose a low-complexity algorithm to solve (P1) sub-optimally  based on the AO technique. Specifically, we alternately optimize one of the $N$ phase shifts in an iterative manner by fixing the other $N-1$ phase shifts, until the convergence is achieved.

For any given phase shift $\bm{\theta}$, it is known that the maximum-ratio transmission (MRT) is the optimal transmit beamforming solution  to (P1) \cite{tse2005fundamentals}, i.e.,
$\bm w^* = \sqrt{ p} \frac{(\bm{h}^H_r\ttheta \bm{G}+\bm{h}^H_d  )^H}{\|\bm{h}^H_r\ttheta \bm{G} +\bm{h}^H_d \|}$, where $p$ denotes the transmit power at the AP. By substituting $\bm w^*$ to (P1),
we obtain the optimal transmit power as $p^* =  \frac{\gamma\sigma^2}{\|\bm{h}^H_r\ttheta \bm{G}+ \bm{h}^H_d\|^2}$. As such, minimizing the transmit power is equivalent to maximizing the channel power gain of the combined channel, i.e.,
\begin{align}\label{secIII:p3}
\text{(P2)}: ~~\max_{\bm{\theta}} ~~~&\|\bm{h}^H_r\ttheta \bm{G}+ \bm{h}^H_d\|^2\\
\mathrm{s.t.}~~~~&\theta_n \in \mathcal{F}, \forall n. \label{SecIII:phaseconstraint}
\end{align}
Let $\bm{\Phi}=\text{diag}(\bm{h}^H_r)\bm{G} \in \mathbb{C}^{N \times M}$, $\A=\bm{\Phi}\bm{\Phi}^H$, and $\bm{\hat h}_d= \bm{\Phi}\bm{h}_d$. Denote by $\A_{n,k}$ and  $\bm{\hat h}_{d,n}$ the $(n,k)$th and $n$th elements in $\A$ and $\bm{\hat h}_d$, respectively.
Then the key to solving (P2) by applying AO lies in the observation that for a given $n\in\{1,...,N\}$,  by fixing $\theta_k$'s, $\forall k\neq n$, the objective function of (P2) is  linear with respect to $e^{j\theta_n}$, which can be written as
{\small\begin{align}\label{obj:phase}
&\! \!\!\!2\mathrm{Re}\!\left\{e^{j\theta_n}\zeta_n\right\} \!+\! \sum_{k\neq n}^{N}\sum_{i\neq n}^{N}\A_{k,i}e^{j(\theta_k-\theta_i)} +C
\end{align}}where $\zeta_n \!=\! \!\sum_{k\neq n}^{N}\A_{n,k}e^{-j\theta_k} \!+\!   \bm{\hat h}_{d,n}\! =\! |\zeta_n |e^{-j\varphi_n}$ and $C\!= \! \A_{n,n} \!+ \!2\mathrm{Re} \Big\{\sum_{k\neq n}^{N}e^{j\theta_k}\bm{\hat h}_{d,k}\Big\}  \! + \!\|\bm{\hat h}_d\|^2$, with $ \mathrm{Re}\{\cdot\}$ denoting the real part of a complex number. Based on \eqref{obj:phase}, it is not difficult to verify that the optimal $n$th phase shift is given by
\begin{align}\label{optimal:phase}
\theta^*_n = \arg \min_{\theta \in \mathcal{F}  } | \theta -\varphi_n|.
\end{align}
By successively  setting the phase shifts of all elements based on (9) in the order from $n=1$ to $n=N$ and repeatedly, the objective value of (P2) is non-decreasing. Since the optimal value of  (P2) is upper-bounded by a finite value, the proposed algorithm is guaranteed to converge. With the converged discrete phase shifts, the optimal transmit power $p^*$ is obtained accordingly.

Note that the above algorithm requires a proper choice of initial discrete phase shifts, which can be obtained by first solving (P1) with the discrete phase-shift constraints \eqref{phase:constraints} replaced by their continuous phase-shift counterparts, i.e.,  $0\leq \theta_n \leq 2\pi, \forall n$ (see [3] for an  algorithm to solve this problem), and  then quantizing the  continuous phase shifts obtained to their nearest points in $\mathcal{F}$ similarly as \eqref{optimal:phase}.  
\vspace{-0.2cm}
\section{PERFORMANCE ANALYSIS}
\vspace{-0.2cm}
In this section, we characterize the scaling law of the average received power at the user with respect to the number of reflecting elements, $N$, as $N\rightarrow \infty$ in an IRS-aided system with discrete phase shifts. For simplicity, we assume $M=1$ with $\G\equiv \bm{g}$ to obtain essential insight. We also assume that the signal received at the user from the AP-user  link can be practically ignored for asymptotically large $N$ since in this case the reflected signal power dominates the total received power at the user.
Thus,  the user's average received  power with $b$-bit phase shifts is approximately given by $P_r(b)\triangleq  \mathbb{E}(|h^H|^2) = \mathbb{E}(|\bm{h}^H_r\ttheta \bm{g}|^2)$ where $\bm{\theta}$ is given by the discrete phase-shift  initialization solution in the proposed algorithm in Section 4.
\begin{proposition}\label{scaling:law}
\emph{\!\!Assume $\bm{h}^H_{r} \!\sim \!\mathcal{CN}(\bm{0},\!\varrho^2_h{\bm I})$ and \! $\bm{g} \!\sim\! \mathcal{CN}(\bm{0},\!\varrho^2_g{\bm I})$.\footnote{  $\mathcal{CN}(\bm{x},{\bm \Sigma})$ denotes the distribution of a circularly symmetric complex Gaussian (CSCG) random vector with mean vector  $\bm{x}$ and covariance matrix ${\bm \Sigma}$, and $\sim$ stands for ``distributed as''.}  As $N\rightarrow \infty$,  we have}
\begin{align}\label{ratio}
\eta (b) \triangleq  \frac{P_r(b)}{P_r(\infty)} = \Big(\frac{2^b}{\pi}\sin\left(\frac{\pi}{2^b}\right)\Big)^2.
\end{align}
\end{proposition}
\begin{proof}
The equivalent channel can be expressed as ${h}^H=   \bm{h}^H_{r} \ttheta\bm{g} =\sum_{n=1}^N |h_{r,n}||g_n|e^{j(\theta_n+\phi_n+\psi_n)}$, where $h^H_{r,n}=|{h}^H_{r}|e^{j\phi_n}$ and $g_n=|g_n|e^{j\psi_n}$ are the corresponding  elements in $\bm{h}^H_{r}$ and $\bm{g}$, respectively. Since $|h_{r,n}|$ and $|g_n|$ are statistically independent and follow Rayleigh distribution with mean values $\sqrt{\pi}\varrho_h/2$ and  $\sqrt{\pi}\varrho_g/2$, respectively, we have $\mathbb{E}(|h_{r,n}||g_n|)=\pi\varrho_h\varrho_g/4$.  Since $\phi_n$ and $\psi_n$ are randomly and uniformly distributed in $[0, 2\pi)$, it follows that $\phi_n+\psi_n$ is uniformly distributed in $[0, 2\pi)$ due to the periodicity over $2\pi$.  As such, the optimal continuous phase shift is given by $\theta^{\star}_n = -(\phi_n +\psi_n)$, $\forall n$ \cite{wu2018IRS}, with the corresponding quantized discrete phase shift denoted by  $\hat{\theta}_n$ which can be obtained similarly as \eqref{optimal:phase}. 
  Define $\bar \theta_n=\hat{\theta}_n- \theta^{\star}_n=\hat{\theta}_n + \phi_n+\psi_n$ as the quantization error. As $\hat{\theta}_n$'s in  $\mathcal{F}$ are equally spaced,   it follows that $\bar \theta_n$'s are independently and uniformly distributed in $[-\pi/2^b, \pi/2^b)$. Then, we have
{\begin{align}
\vspace{-2mm}
&\!\!\!\mathbb{E}(|{h}^H|^2)\!  =\!  \mathbb{E}\left(\left|\sum_{n=1}^N |h_{r,n}||g_n|e^{j{\bar \theta_n}}  \right|^2 \right ) \! \!=\!  \mathbb{E} \!\left(   \sum_{n=1}^N |h_{r,n}|^2|g_n|^2 \right. \nonumber \\
~&~~~~~~~~~~~~\quad~\left.   +        \sum_{n=1}^N \sum_{i\neq n}^N |h_{r,n}||g_n||h_{r,i}||g_i|e^{j{\bar \theta_n}-j{\bar \theta_i}}   \right).
\vspace{-0.3cm}
\end{align}}Note that $h_{r,n}$, $g_n$, and $e^{j\bar \theta_n}$ are independent with each other, with $\mathbb{E}\left(   \sum_{n=1}^N |h_{r,n}|^2|g_n|^2\right)= N\varrho^2_h\varrho^2_g$ and $\mathbb{E}(e^{j\bar \theta_n})=\mathbb{E}(e^{-j\bar \theta_n})= 2^b/\pi\sin\left(\pi/2^b\right)$. It then follows that
\begin{align}\label{eq:pow}
\!\!\!\!P_r(b)&\! =\!N\varrho^2_h\varrho^2_g + N(N-1)\frac{\pi^2\varrho^2_h\varrho^2_g}{16}\Big(\frac{2^b}{\pi}\sin\left(\frac{\pi}{2^b}\right)\Big)^2.
\end{align}
For $b\geq 1$, is is not difficult to verify that $2^b/\pi\sin\left(\pi/2^b\right)$ increases with $b$  monotonically and approaches to $1$  when $b \rightarrow \infty$ (i.e., continuous phase shifts without quantization).  As a result, the ratio of $P_r(b)$ and $P_r(\infty)$ is given by \eqref{ratio} when $N\rightarrow \infty$, which completes the proof.
\end{proof}
Proposition \ref{scaling:law} provides a quantitative measure of  the user received power loss with discrete phase shifts as compared to the ideal case of continuous phase shifts. It is observed that  as $N\rightarrow \infty$, the power ratio $\eta (b)$ depends only on the resolution of phase shifters, but is regardless of $N$. This result implies that even with a practical IRS with discrete phase shifts, the same asymptotic squared power gain of $\mathcal{O}(N^2)$ as that with  continuous phase shifts can be achieved (see  \eqref{eq:pow} with $N\rightarrow \infty$).  As such,  the design of  IRS hardware and control module can be greatly simplified by using discrete phase shifters, without compromising the performance in the large-$N$ regime.  Since $\eta (1)  =0.4053 $, $\eta (2)= 0.8106$,  and $\eta (3)=0.9496$,   using 2 or 3-bit phase shifters is practically  sufficient to achieve close-to-optimal  performance. In general, to achieve a given received power at the user,  there exists an interesting  trade-off between the number of reflecting elements and the resolution of phase shifters used at the IRS.

\vspace{-0.3cm}
\section{Simulation Results}
\vspace{-0.3cm}
 We consider a uniform linear array (ULA) at the AP and a uniform rectangular array (URA) at the IRS, respectively.
   The signal attenuation at a reference distance of 1 meter (m) is set as 30 dB  for all channels. Since the IRS can be practically deployed to avoid blockage between the AP and its covered area in which the user of our interest is usually located, the pathloss exponent of the AP-IRS channel is set to be $2.2$, which is lower than that of the IRS-user channel $(2.8)$, and that of the  AP-user channel $(3.4)$.
To account for small-scale fading, we assume Rayleigh fading for all the channels involved.
 Other  parameters are set as follows:   $\sigma^2=-80$\,dBm, $\gamma=20$\,dB, and $M=5$.

The AP and IRS are assumed to be  located $d_0= 50$ m apart and the user lies on a horizontal line that is in parallel to the one that connects them, with the vertical distance between these two lines equal to  $d_v = 2$ m. Denote the horizontal distance between the AP and  user  by $d$ m. The AP-user and IRS-user link distances are then given by $d_1= \sqrt{d^2+d_v^2}$ and $d_2= \sqrt{(d_0-d)^2+d_v^2}$, respectively.
 By varying the value of $d$, we examine the minimum transmit power required for serving the user with the given SNR target. We compare the following schemes: 1) Lower bound: solving (P1) with $b\rightarrow \infty$ by using semidefinite programming in \cite{wu2018IRS};  2) Exhaustive search with 1-bit IRS: solving (P1) by searching all possible combinations of binary $(K=2)$  phase shifts; 3) AO with 1-bit IRS: using the AO algorithm in Section 4;  4) Initialization scheme with 1-bit IRS in Section 4; 5)  The scheme without using the IRS by setting  $\bm{w} = \sqrt{\gamma \sigma^2}{\bm{h}_d}/{\|\bm{h}_d\|^2}$.
In Fig. \ref{simulation:distance}, we compare  the  transmit power required at the AP by the above schemes versus the AP-user horizontal distance.
First,  it  is observed that the required transmit power of using 1-bit phase shifters is significantly lower than that without the IRS when the user locates in the vicinity of the IRS. This demonstrates the practical usefulness of IRS in creating a ``signal hotspot'' even with very coarse and low-cost phase shifters.
 Moreover, one can observe that the 1-bit phase shifters suffer power loss compared to the case  with continuous phase shifts. This is expected since due to coarse discrete phase shifts,  the  multi-path signals  from the AP including those reflected and non-reflected by the IRS  cannot be perfectly aligned in phase at the receiver, thus resulting in a power loss.
Finally,  it is observed that compared to the exhaustive search scheme, the proposed AO algorithm and the initialization scheme both achieve near-optimal  performance.

To validate the analytical result in Proposition \ref{scaling:law}, we further compare in Fig.  \!\ref{simulation:N} the AP transmit power  versus the number of reflecting elements $N$ at the IRS when $d=50$ m. In particular, we consider both $b=1$ and $b=2$ for discrete phase shifts at the IRS.  From Fig. \!\ref{simulation:N}, it is  observed that as $N$ increases, the performance gap between the proposed scheme (for both $b=1$ and $b=2$) and the lower bound first increases and then approaches a constant that is determined by $\eta(b)$ given in \eqref{ratio} (i.e., $\eta (1)  =-3.9224$ dB  and $\eta (2)=-0.9224$ dB). This is expected since when $N$ is moderate,  the signal power of the AP-user link is comparable to that of the IRS-user link, thus the misalignment of multi-path signals due to discrete phase shifts becomes more pronounced with increasing $N$. However, when $N$ is sufficiently large such that the reflected signal power by the IRS  dominates the total received power at the user, the performance loss arising from the phase quantization error converges to that by the asymptotic analysis given in Proposition  \ref{scaling:law}.
In addition, one can observe that in this case the gain achieved by the AO scheme over the initialization scheme is more evident with $b=1$ compared to $b=2$.

 \begin{figure}[!t]
\centering
\includegraphics[width=0.35\textwidth]{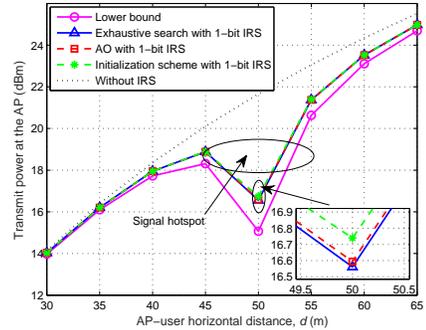}
\vspace{-0.4cm}
\caption{AP transmit power versus AP-user horizontal distance. } \label{simulation:distance} \vspace{-4.5mm}
\end{figure}

\begin{figure}[!t]
\centering
\includegraphics[width=0.35\textwidth]{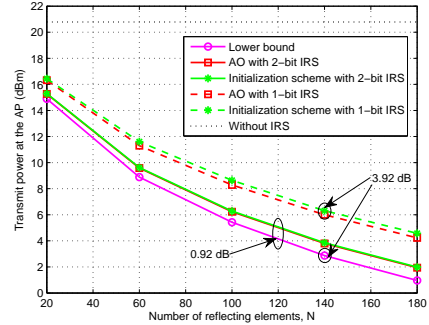}
\vspace{-0.4cm}
\caption{AP transmit power versus the number of  elements.  } \label{simulation:N} \vspace{-5mm}
\end{figure}

\vspace{-0.4cm}
\section{Conclusion}
\vspace{-0.2cm}
In this paper, we studied the beamforming optimization for IRS-enhanced wireless communication under discrete phase-shift constraints at the IRS.  Specifically, the continuous transmit beamforming at the AP and the discrete reflect beamforming at the IRS  were jointly optimized to minimize the transmit power at the AP under the given user SNR target. We proposed an efficient AO-based algorithm to solve this problem, which was shown to achieve near-optimal performance. Furthermore,   we qualitatively analyzed the performance loss caused by using IRS with discrete phase shifts as compared to the ideal case with continuous phase shifts when the number of reflecting elements becomes  asymptotically large. Interestingly, it was shown that even using IRS with 1-bit phase shifters is able to achieve the squared power gain as in the case with continuous phase shifts. Simulation results demonstrated the transmit power saving achieved by using IRS with discrete phase shifts as compared to the case without IRS. In addition, it was shown that obtaining discrete phase shifts by directly quantizing  the continuous phase solution already achieves quite good performance, while the gain by the additional AO is only marginal.


\bibliographystyle{IEEEtran}
\bibliography{IEEEabrv,mybib}

\end{document}